\documentclass[11pt]{llncs}
\usepackage{fullpage}
\usepackage{times}
\usepackage{amssymb}
\usepackage{amsmath}
\usepackage{color}
\usepackage{xspace}
\usepackage{cite}
\usepackage{subfigure}
\usepackage{epsfig}
\usepackage{psfrag}       %latex text in .eps files
\usepackage[ruled,vlined,linesnumbered]{algorithm2e}%noend
\usepackage{url}

\newcommand{\thmlabel}[1]     {\label{thm:#1}}
\newcommand{\thmref}[1]     {Theorem~\ref{thm:#1}}
\renewcommand{\eqref}[1]      {(\ref{eq:#1})}
\newcommand{\eqlabel}[1]     {\label{eq:#1}}
\newcommand{\parhead}[1]{\noindent{\textbf{#1.}\xspace}}
\newcommand{\myboldmath}{}%\boldmath doesn't work with latex
\newcommand{\defn}[1]           {{\textit{\textbf{\myboldmath #1}}}}
\newcommand{\mmb}{Multiple-Message Broadcast\xspace}

\newcommand{\RNs}{Radio Networks\xspace}
\newcommand{\RN}{{Radio Network}\xspace}
\newcommand{\footnotenonumber}[1]{{\def\thempfn{}\footnotetext{#1}}}

\begin{document}

\title{Dynamic \mmb:\\ Bounding Throughput in the Affectance Model}

\titlerunning{Dynamic \mmb}

\author{
Dariusz~R.~Kowalski\inst{1}
\and
Miguel~A.~Mosteiro\inst{2}
\and
Kevin~Zaki\inst{2}
}

\authorrunning{D.~R.~Kowalski, M.~A.~Mosteiro, and K.~Zaki}

\institute{        
       Univ. of Liverpool, 
       Dept. of Computer Science,
       Liverpool, UK.\\
       \email{D.Kowalski@liv.ac.uk}
\and
       Kean University,
       Dept. of Computer Science,
       Union, NJ.\\
       \email{\{mmosteir,zakike\}@kean.edu}
}

\maketitle

\begin{abstract}
We study a dynamic version of the \mmb problem, where packets are continuously injected in 
network nodes for dissemination throughout the network. 
Our performance metric is the ratio of the throughput of such protocol against the optimal one, 
for any sufficiently long period of time since startup. We present and analyze a dynamic \mmb protocol that works under an affectance model, which parameterizes the interference that other nodes introduce in the communication between a given pair of nodes. 
As an algorithmic tool, we develop an efficient algorithm to schedule a broadcast along a BFS tree under the affectance model. 
To provide a rigorous and accurate analysis, we define two novel network
characteristics based on the network topology and the affectance function.
The combination of these characteristics
influence the performance of broadcasting with affectance (modulo a logarithmic function).
We also carry out simulations of our protocol under affectance. To the best of our knowledge, this is the first dynamic \mmb protocol that provides throughput guarantees for continuous injection of messages and works under the affectance model. 
\end{abstract}

\footnotenonumber{A preliminary version of this work has appeared in~\cite{KMR_fomc14}. The differences with respect to that version are detailed in Section~\ref{sec:diff} in the Appendix.}

%!TEX root = ./KMZ_mmb.tex

\section{Introduction}

We study the dynamic \mmb problem in wireless networks under the {\em affectance} model. 
This model subsumes many communication-interference models studied in the literature, such as 
\RN (cf.,~\cite{ChlamtacK87}) and models based on the Signal to Interference and Noise Ratio (SINR)
(cf.~\cite{HWaff,SRSSINRdomSet}).
The notion of affectance was first introduced in~\cite{HWaff} in the context of link scheduling
in the more restricted SINR model of wireless networks, in an attempt 
to formalize the combination of interferences from
a subset of links to a selected link under the SINR model.
%the SINR communication model. 
Later on, other realizations of affectance were defined and abstracted as an independent model
of interference in wireless networks~\cite{Kaff, KVaff}. 
The conceptual idea of this model is to parameterize the interference that transmitting nodes 
introduce in the communication between a given pair of nodes. 
%In this work, we realize the affectance model with a matrix $A$ of size $n\times m$, where $n$ is the number of nodes and $m$ the number of communication links in the network. Then, $a_{i,j}$ quantifies the interference that node $i$ introduces in the communication through link $j$.

\parhead{Our results}
In the {\em dynamic \mmb} problem considered in this work, packets arrive at nodes in an online fashion and
need to be delivered to all nodes in the network. 
We are interested in the \emph{throughput}, i.e., the number of packets delivered in a given period of time.
In particular, we measure {\em competitive throughput} of 
{\em deterministic distributed algorithms} for the dynamic \mmb problem.  
We analyse our algorithms in the (general) affectance model, in which there is a given undirected 
communication graph $G$ of $n$ nodes and diameter $D$, together with the affectance function $a(\cdot)$
of nodes of distance at least $2$ on each of the communication links.
The affectance function has a {\em degradation parameter} $\alpha$, being a distance after which
the affectance is negligible.
Our contribution is two fold.

First, we introduce new model characteristics --- based on the underlying communication network
and the affectance function ---
called {\em maximum average tree-layer affectance} (denoted by $K$)
and {\em maximum path affectance} (denoted by $M$), see Section~\ref{s:prelim} for the definitions,
and show how they influence the time complexity of broadcast. More precisely,
if one uses a BFS tree %, called GBST (cf.,~\cite{GPQ:gossipTree}),
that minimizes the product $M\cdot (K+M/\log n)$\footnote{%
Throughout, we denote $\log_2$ simply as $\log$, unless otherwise stated.%
} of the two above characteristics, 
then a single broadcast can be done in time $D+O(M(K+M/\log n)\log^2 n)$,
%where $n$ is the number of nodes and $D$ is the diameter of the network,
cf., Corollary~\ref{c:broadcast} in Section~\ref{sec:BT}. 

Second, we extend this method of analysis to a dynamic packet arrival model and
the \mmb problem, and design a new algorithm reaching competitive throughput of
$\Omega(1/(\alpha K \log n))$. In particular, in the \RN model it implies a competitive throughput of
$\Omega(1/(\log^2 n))$. For details, see Section~\ref{section:prot}.
Our deterministic results are existential, that is, 
we show the existence of a deterministic schedule 
%providing 
by applying a probabilistic argument to 
a protocol that includes a randomized subroutine for layer to layer dissemination. 
%To use this protocol as Las Vegas, acknowledgement from the channel is required, and the same throughput can be attained w.h.p. Otherwise, the protocol is Monte Carlo. 
Given that we measure competitive throughput in the limit, preprocessing (communication infrastructure setup, topology information dissemination, etc.) can be carried out initially without asymptotic impact.
Thus, the protocol presented is distributed, and it works for \emph{every} network after learning its topology.
The protocol can also be applied to mobile networks, if the movement is slow enough to recompute the structure.
Our rigorous asymptotic analysis is further complemented by simulations
under the affectance model, c.f., Section~\ref{s:simul}.

To the best of our knowledge, ours is the first work on the dynamic \mmb problem in wireless
networks under the general affectance model.

\parhead{Previous and related work}
There is a rich history of research on broadcasting dynamically arriving packets on 
a {\em single-hop radio network}, also called a {\em multiple access channel}.
Most of the research focused on {\em stochastic arrivals}, cf., a survey by Chlebus~\cite{Chlebus-01}.
In the remainder of this paragraph, we focus on the on-line adversarial packet arrival setting.
Bender et al. \cite{BenderFHKL05} studied stability, understood as throughput being
not smaller than the packet arrival rate, of randomized backoff protocols on multiple access channels in 
%adversarial settings in 
the {\em queue-free model}, in which every packet is handled independently as if it has been a standalone station
(thus avoiding queuing problems).
%Stability was defined to mean that output rate was as large as injection rate. 
%They showed that exponential backoff was unstable for rates $\rho\ge c \lg\lg n/\lg n$, for a sufficiently large constant~$c$.
%Among positive results, they showed that log-log iterated backoff was stable when rates satisfied $\rho\le c/(\lg\lg n\lg n)$,  and that exponential backoff was stable for rates $\rho \le c/\lg n$, for a sufficiently small constant~$c$.
%Some aspects of dynamic selection on multiple access channels by deterministic protocols were considered by 
Kowalski~\cite{Kowalski05} considered a dynamic broadcast on the channel in the setting where packets 
could be combined in a single message, which again avoids various 
important issues related with queuing.
%His approach differed from the one in this paper in that all packets in the queue at a station could be combined and sent in a single transmission; moreover, protocols were restricted to be acknowledgment based only.
Anantharamu et al. \cite{AnantharamuCKR-INFOCOM10} studied packet latency of deterministic 
dynamic broadcast protocols for arrival rates {\em smaller} than~$1$.
%Packet latency investigated in~\cite{AnantharamuCKR-INFOCOM10} was considered as function of injection rate and size of the system, with these two parameters treated as independent. 
%This work differs from~\cite{AnantharamuCKR-INFOCOM10} in that we want our protocols to achieve fair latency for as large injection rate as possible in a given system, which results in specifying injection rates as functions of the number of stations.
%
Stability, understood as bounded queues, of dynamic deterministic broadcast on multiple access channels against 
adversaries bounded by arrival rate~$1$ was studied by Chlebus et al. \cite{ChlebusKR09}, and for arrival rates smaller than $1$ by Chlebus et al. \cite{ChlebusKR12}.
In particular, in \cite{ChlebusKR09} a protocol Move-big-to-front (MBTF) was designed,
achieving stability but not fairness (as both these properties are impossible to achieve simultaneously);
we use this algorithm as a subroutine in our dynamic \mmb protocol.
A follow-up work~\cite{BienkowskiJKK12} delivered a distributed online algorithm {\sc Scat} and showed that
it could be only by a linear factor worse, in terms of the buffer size, than any offline solution against
any arrival pattern.

In {\em multi-hop} \RNs, the previous research concentrated on time complexity of single instances (i.e., from a single source) of broadcast and multi-message broadcast.
For directed networks, the best deterministic solution is a combination of the $O(n\log n \log\log n)$-time algorithm by
De Marco~\cite{DeMarco08} and the $O(n\log^2 D)$-time algorithm by Czumaj and Rytter~\cite{CzumajR03}. In undirected networks, the best up to date deterministic broadcast in $O(n\log(n/D))$ rounds was given by Kowalski~\cite{Kowalski05}.
The lower bounds for deterministic broadcast in directed and undirected radio networks are
$\Omega(n\log(n/D))$~\cite{ClementiMS01} and $\Omega(n\log_D n)$~\cite{KowalskiP05}, respectively.
Deterministic multi-message broadcast, group communication and gossip were also considered
(again, in a single instance). 
Chlebus et al.~\cite{ChlebusKPR11} showed a $O(k\log^3 n+n\log^4 n)$ time
deterministic multi-broadcast algorithm for $k$ packets in undirected radio networks.
Single broadcast can be done optimally in $\Theta(D\log(n/D)+\log^2 n)$, as proved
in~\cite{AlonBLP91,KushilevitzM98} (lower bounds) and in~\cite{CzumajR03,KowalskiP05} (matching upper bound). 
Bar-Yehuda et al.~\cite{BII93}, and recently 
Khabbazian and Kowalski~\cite{KhabbazianK11} and 
Ghaffari et al.~\cite{GHK:randBroad}, 
studied randomized multi-broadcast protocols; the best results obtained for $k$-sources single-instance multi-broadcast is the amortized $O(\log \Delta)$ rounds per packet w.h.p. in~\cite{KhabbazianK11}, where $\Delta$ is the maximum node degree, and $O(D+k\log n+\log^2 n)$ w.h.p. to broadcast the $k$ packets, for settings with known topology in~\cite{GHK:randBroad}.
For the same problem, Ghaffari et al. showed a throughput upper bound of $O(1/\log n)$ for any algorithm in~\cite{GHK:throughputLB}. Although this bound is worst-case, it can be compared with our $1/O(\alpha K \log n)$ that applies even under affectance.

Chlebus et al.~\cite{ChlebusKR09} gave various deterministic and randomized algorithms
for group communication, all of them being only a small polylogarithm away of the corresponding
lower bounds on time complexity.

In the {\em SINR model}, single-hop instances of 
broadcast in the ad-hoc setting were studied 
by Jurdzinski et al.~\cite{JurdzinskiKRS13,JurdzinskiKS13} and Daum et al.~\cite{DaumGKN13}, who gave several deterministic and
randomized algorithms working in time proportional to the diameter multiplied by a polylogarithmic
factor of some model parameters. In the SINR model with restricted sensitivity, so called
weak-sensitivity device model, Jurdzinski and Kowalski~\cite{JK-DISC-12} designed an algorithm
spanning an efficient backbone sub-network, that might be used for efficient implementation
of multi-broadcast.

The {\em generalized affectance} model was introduced and used only in the context of one-hop
communication, more specifically, to link scheduling by Kesselheim~\cite{Kaff}. 
He also showed how to use it for dynamic link scheduling in batches.
This model was inspired by the affectance parameter introduced in the more restricted SINR setting~\cite{HWaff}.
They give a characteristic of a set of links, based on affectance, that influence
the time of successful scheduling these links under the SINR model.
In our paper, we generalize this characteristic, called the maximum average tree-layer affectance,
to be applicable to multi-hop communication tasks
such as broadcast, together with another characteristic, called the maximum path affectance.
For details see Section~\ref{s:prelim}.

%!TEX root = ./KMZ_mmb.tex

\section{Preliminaries}
\label{s:prelim}

\paragraph{Model.}
We study a model of network consisting of $n$ \defn{nodes}, where communication is carried out through radio \defn{transmissions} in a shared channel. 
Time is discretized in a sequence of time slots $1,2,\dots$, which we call the \defn{global time}. 
The network is modeled by the underlying \defn{connectivity graph} $G=\{V,E\}$, where $V$ is the set of nodes and $E$ the set of \defn{links} among nodes. A link $\ell\in E$ between two nodes $u,v\in V$ is the ordered pair $(u,v)$ modeling that a transmission from $u$ may be received by $v$. The network is assumed to be connected but \defn{multihop}. That is, any pair of nodes may communicate, possibly through multiple hops. 

Messages to be broadcast to the network through radio transmissions are called \defn{packets}.
Packets are \defn{injected} at nodes at the beginning of time slots, and each time slot is long enough to transmit a packet to a neighboring node. 
Any given node can either transmit or listen (in order to receive, if possible) in a time slot, but not both.

Interference on a link due to transmissions from other nodes is modeled as \defn{affectance}. 
We use a model of affectance that subsumes other interference models~\footnote{In preliminary work~\cite{KMR_fomc14}, we studied a different model of affectance. The details are included in Section~\ref{sec:diff} in the Appendix.},
such as the \RN model~\cite{ChlamtacK87} and the SINR model~\cite{HWaff}.
Specifically, we realize affectance as a matrix $A$ of size $|V|\times|E|$ where $A(u,(v,w))$ quantifies the interference that a transmitting node $u\in V$ introduces to the communication through link $(v,w)\in E$. 
We do not restrict ourselves to any particular affectance function, as long as its effect is additive. 
That is, denoting $a_{V'}((v,w))$ as the affectance of a set of nodes $V'\subseteq V$ on a link $(v,w)\in E$, 
and $a_{V'}(E')$ as the affectance of a set of nodes $V'\subseteq V$ on a set of links $E'\subseteq E$, it is
\begin{align*}
a_{V'}((v,w)) &= \sum_{\substack{u\in V'\\u\neq v}}A(u,(v,w))\\
a_{V'}(E') &= \sum_{(v,w)\in E'}a_{V'}((v,w)).
\end{align*}

Under the affectance model, we define a \defn{successful transmission} as follows.
For any pair of nodes $u,v \in V$ such that $(u,v)\in E$, a transmission from $u$ is received at $v$ in a time slot $t$ if and only if: $u$ transmits and $v$ listens in time slot $t$, 
%and $\forall w\in V\setminus \{u\}$ such that $\{w,v\}\in E$, $w$ does not transmit in $t$, 
and $a_{\mathcal{T} (t)}((u,v))< 1$, where $\mathcal{T} (t)\subseteq V$ is the set of nodes transmitting in time slot $t$ (notice that the definition of $a$ does not include the affectance of $u$ on $(u,v)$). 
The event of a non-successful transmission, that is when the affectance is at least $1$, is called a \defn{collision}. 
We assume that a node listening to the channel cannot distinguish between a collision and background noise present in the channel in absence of transmissions. 

The affectance model defined subsumes other interference models. For instance, for the \RN model, the affectance matrix is 
\begin{align*}
A(w,(u,v)) &=
\left\{ \begin{array}{ll}
 0 & \textrm{if $w=u$ or ($(w,v)\notin E$ and $w\neq v$)},\\
 1 & \textrm{otherwise}.
  \end{array} \right.
\end{align*}
On the other hand, for the SINR model in~\cite{HWaff}, the affectance matrix is 
\begin{align*}
A(w,(u,v)) &= 
\left\{ \begin{array}{ll}
 0 & \textrm{if $w=u$},\\
 \frac{P/d_{wv}^\alpha}{P/(\beta' d_{uv}^\alpha)-N} & \textrm{otherwise}.
  \end{array} \right.
\end{align*}
Where $P$ is the transmission power level, $N$ is the background noise, $\beta'$ denotes an upper bound on the signal to interference-plus-noise ratio such that a message cannot be successfully received, $d_{uv}$ is the euclidean distance between nodes $u$ and $v$, and $\alpha$ denotes the path-loss exponent. (Refer to Section~\ref{s:mapping} in the Appendix for a proof.)

\paragraph{Communication task.}
Under the above model, we study the \defn{\mmb} problem defined as follows.
Starting at time slot $1$, packets are dynamically injected by an exogenous entity into some of the network nodes, called \defn{source nodes}. The computing task is to disseminate those injected packets throughout the network.
The set of all source nodes is denoted as $S\subseteq V$.
After a packet has been received by all the nodes in the network, we say that the packet was \defn{delivered}.
The injections are adversarial, that is, packets can be injected at any time slot at any source node, but the injections are limited to be feasible. 
We say that an injection is \defn{feasible} if there exists an optimal algorithm OPT such that the \defn{latency} (i.e., the time elapsed from injection to delivery) of each packet is bounded for OPT. 
Given that at most one packet may be received by a node in each time slot, and that all nodes must receive the packet to be delivered, feasibility limits the adversarial injection rate to at most $1$ packet per time slot injected in the network.
The goal is to find a \defn{broadcasting schedule}, that is, a temporal sequence of transmit/not-transmit states for each node, so that packets are delivered. We denote the period of time since a packet is transmitted from the source until its delivery as the \defn{length} of the schedule.

\paragraph{Performance metric.}
%We do not want to limit the packet arrival rate (as in AQT) as long as a bounded latency is achievable. Instead, we evaluate the \defn{ratio} of the performance of an algorithm ALG against an optimal algorithm OPT. We would like to bound this ratio. We thought first about measuring the latency ratio. However, we would have to study bad cases when bursts of arrivals occur. Preliminary, we observed that in such cases the ratio ALG/OPT is unbounded. 
%Provide examples to motivate the choice of a different performance metric. 
%For one-hop networks, we can motivate citing the result in~\cite{ChlebusKR09}:
%if on average one packet per round is injected, and the OPT has always no more than 2 packets
%in the system, no distributed protocol has bounded latency; more precisely: no protocol
%is both stable (bounded number of packets in the system) and fair (every packet
%is eventually successfully transmitted). For multi-hop networks we understand the same result
%holds as a natural extension of the single hop model.
%Motivated by the above, 
We evaluate the \defn{ratio} of the performance of a distributed online algorithm ALG against an optimal algorithm OPT. 
For one hop networks it is known~\cite{ChlebusKR09} that no protocol is both \defn{stable} 
(i.e., bounded number of packets in the system at any time) and \defn{fair} (i.e., every packet is eventually delivered). For multihop networks the same result holds as a natural extension of the single hop model. 
Thus, instead of further limiting the adversary (beyond feasibility) to achieve either stability or bounded latency, our goal is to prove a lower bound on the \defn{competitive throughput}, for any sufficiently long prefix of time slots since global time $1$. 
Specifically, we want to prove that there exists a function $f$, possibly depending on network parameters, such that 
$$\lim_{t\to\infty} \frac{d_{ALG}(t)}{d_{OPT}(t)} \in \Omega(f),$$ 
where $d_X(t)$ is the number of packets delivered to all nodes by algorithm $X$ until time slot $t$.

\paragraph{Network characterization.}

We characterize a network by its \defn{affectance degradation distance}, which is the number of hops $\alpha$ 
such that the affectance of nodes of distance at least $\alpha$ to a given link is ``negligible'', that is, zero. 
Additionally, we characterize the network with two measures of affectance based on broadcast trees~\footnote{The second characterization was presented differently in the conference version of this work. The details are included in Section~\ref{sec:diff} in the Appendix.}, as follows. 
Given a network with a set of nodes $V$ including a source node $s$, 
consider a Breadth First Search (BFS) tree $T$ rooted at $s$. 
For any $d=0,1,2,\dots$, let $V_d(T)$ be the set of all nodes at (shortest) distance $d$ from $s$.
Based on this tree, we define the \defn{maximum average tree-layer affectance} as
$$K(T,s)=\max_d \max_{V'\subseteq V_d(T)} \frac{1}{|L(V')|} a_{V'} (L(V')) \ ,$$
where $L(V')$ is the set of tree links between $V'$ and nodes at distance $d+1$ of the source.
Intuitively, $K(T,s)$ indicates what might be the worst affectance to overcome when
trying to broadcast from one layer of $T$ to another.
We also define the \defn{maximum path affectance} as
%$$M(T,s)=\max_{p\in P(T)}\sum_{\substack{(u,v)\in p\\d(v)=d(u)+1}} a_{V_{d(u)}(T)}((u,v)) \ ,$$
$$M(T,s)=\max_{p\in P(T)}\sum_{(u,v)\in p} a_{V_{d(u)}(T)}((u,v)) \ ,$$
where 
$d(u)$ is the distance from node $u$ to $s$,
and $P(T)$ is the set of paths root-to-leaf in $T$ (i.e., a set of sets of links), where a \defn{path root-to-leaf} is the standard notion of a set of links $\{(s,x_1),(x_1,x_2),(x_2,x_3),\dots,(x_{k-1},x_k)\}$ such that $x_k$ is a leaf. 
Intuitively, $M(T,s)$ indicates what is the worst affectance when trying to pipeline packets through a path down the tree.
In the rest of the paper, the specific tree and source node will be omitted when clear from context.

%!TEX root = ./KMZ_mmb.tex

\section{A Broadcast Tree}
\label{sec:BT}
 
In this section, we show a broadcasting schedule that, under the affectance model, disseminates a packet held at a source node to all other nodes. The schedule is defined constructively with a protocol that uses randomization, thus providing only stochastic guarantees. Given that the protocol is Las Vegas, the construction also proves the existence of a deterministic broadcasting schedule. 

First, we detail the construction of a ranked tree spanning the network rooted at the source node
that will be used to define the broadcasting schedule that we detail afterwards. 
The construction borrows the idea in~\cite{GPQ:gossipTree} of defining some nodes as \emph{fast} and others as \emph{slow} based on rank. 
However, our rank is a consequence of affectance rather than \RN collisions, and it is defined to schedule transmissions downwards the tree only, rather than both directions.
Moreover, our definition of the fast node sets, the slot reservation, and the contention resolution protocol are also different.
The following notation will be used.

Given a tree $T(s)\subseteq E$ rooted at $s\in V$, spanning a set of network nodes $V$ with set of links $E$, 
let $d(v)$ be the \defn{distance} in hops from a node $v\in V$ to the root of $T(s)$,
let $p(\ell)$ and $c(\ell)$ be the parent and child nodes of link $\ell\in T(s)$ respectively,
and let $D(T(s))$ be the 
maximum distance in $T(s)$ from any node to the root $s$.
Additionally, a \defn{rank} (a number in $\mathbb{N}$) will be assigned to each node.
Let $r(u)$ be the rank of node $u\in V$,
let $R(T(s))$ be the maximum rank in the tree,
and let $F_d^r (T(s))= \{u | u\in V_d \land r(u)=r \land \exists v\in V_{d+1} : ((u,v)\in T(s) \land r(v)=r)\}$, 
that is, the set of nodes of rank $r$ at distance $d$ from the root that have a child with the same rank.
Let a node $v\in V$ be called \defn{fast} if it belongs to the set $F_{d(v)}^{r(v)}(T(s))$, and \defn{slow} otherwise.
The sets $F_d^r$ are called \defn{fast node sets} whereas the set containing all slow nodes is called \defn{slow node set}.
In the above notation, the specific tree parameter and/or source node will be omitted when clear from the context. 

Given a graph $G$ and a source node $s\in S$, consider the following construction of a 
\defn{Low-Affectance Broadcast Spanning Tree (LABST)}.
Let $T_{min}$ be the BFS tree that minimizes the following polynomial on the affectance measures. Letting $\mathcal{T}$ be the class of all BFS trees with source $s$, it is 
$$\forall T\in \mathcal{T} : M(T_{\min},s) \left(  \frac{M(T_{\min},s)}{\log n} + K(T_{\min},s) \right) \leq M(T,s) \left(  \frac{M(T,s)}{\log n} + K(T,s) \right).$$
Then, using Algorithm~\ref{alg:prunaff}, we define a rank on each node, that is, transform $T_{\min}$ into a LABST $T$, to avoid links between nodes of the same rank with big affectance. 

In brief, the transformation is the following (refer to Algorithm~\ref{alg:prunaff}).
Initially, the rank of all nodes is set to~$1$, and the fast node sets are initialized in Lines~\ref{initsetbegins} to ~\ref{initsetends}. Then, for each distance $d$ upwards the tree, two phases are executed as follows.

In a first phase (Lines~\ref{updaterankdbegins} to~\ref{updaterankdends}), the rank of all nodes at distance $d$ is updated if necessary. That is, for each increasing rank $r$, and for each link $\ell$ such that the parent node $u$ is located at distance $d$ (hence, child node at distance $d+1$) and parent and child nodes have rank $r$, check the affectance on link $\ell$ from other rank-$r$-distance-$d$ nodes with a rank-$r$ child. If this affectance is at least $1$, increase the rank of $u$, and remove $u$ from the fast node set $F_d^r$ since its rank is not $r$ anymore. Notice that $u$ had the maximum rank among its children because it was in a fast node set, but now has a rank bigger than any of its children. Hence, $u$ is now in the slow node set.

In the second phase (Lines~\ref{updaterankupwardsbegins} to~\ref{updaterankupwardsends}),
the rank of all the ancestors (distance $<d$) is updated so that the rank of a node at distance $d$ equals the maximum rank (not necessarily unique) among its children (that is, ranks are monotonically non-decreasing upwards). 
While computing the rank, all fast node sets $F_d^r$ are also updated.

\begin{algorithm}[htbp]
\dontprintsemicolon
\label{alg:prunaff}
\caption{LABST construction.}
$T\leftarrow T_{\min}$\;
\ForEach{distance $d=0,1,2,\dots,D(T)$\nllabel{initsetbegins}}{ 
	\ForEach(\tcp*[f]{$R(T)\leq D(T)+1$}){rank $r=1,2,\dots,D(T)+1$}{
		$F_{d}^r\leftarrow\emptyset$\;
	}
}
\ForEach{$u\in V$}{ 
	$r(u)\leftarrow 1$\tcp*[f]{initially, all nodes have rank $1$}\nllabel{line:initrank}\;
	\If{$u$ is not a leaf in $T$}{$F_{d(u)}^1 \leftarrow F_{d(u)}^1 \cup \{u\}$\nllabel{initsetends}\;}
}
\ForEach{distance $d=D(T)-1,\dots,2,1,0$\nllabel{line:upwards}}{ 
\ForEach{rank $r=1,2,\dots,D(T)+1$\nllabel{updaterankdbegins}}{ 
\ForEach{node $u$ such that $u\in F_d^r$}{ 
\ForEach{link $\ell$ such that $p(\ell)=u$}{ 
\If{$u\in F_d^r$ {\bf and} $a_{F_d^r}(\ell)\geq1$\nllabel{line:reduce}}{
	$r(u)\leftarrow r(c(\ell))+1$\nllabel{line:increaserank}\;
	$F_d^r\leftarrow F_d^r\setminus\{u\}$\tcp*[f]{$u$ is now slow}\nllabel{line:toslow}\nllabel{updaterankdends}\;
}
}
}
}
\ForEach{distance $d'=d-1,\dots,2,1,0$\nllabel{updaterankupwardsbegins}}{
	\ForEach(\tcp*[f]{update ranks and sets}){node $u\in V_{d'}$}{
		$F_{d'}^{r(u)}\leftarrow F_{d'}^{r(u)}\setminus\{u\}$\;
		$r_{\max} \leftarrow \max_{\substack{w\in V_{d'+1}\\(u,w)\in T}} r(w)$\tcp*[f]{max rank of children of $u$}\;
		\If{$r(u)\leq r_{\max}$}{
			$r(u)\leftarrow r_{\max}$\;
			$F_{d'}^{r(u)}\leftarrow F_{d'}^{r(u)}\cup\{u\}$\nllabel{updaterankupwardsends}\;
		}
	}
}
}
\end{algorithm}

The broadcasting schedule is defined using the LABST $T$ obtained. Being a radio-broadcast network, transmissions might be received using other links or time slots, but the LABST and broadcasting schedule defined provide the communication guarantees.
Each node follows certain broadcasting schedule, but using only time slots reserved for itself.
Then, for each node $v\in V$, 
if $v$ is fast, it uses each time slot $t$ such that $t \equiv d(v)+2h(R(T) - r(v)) \pmod{2hR(T)}$, 
where $h=\max\{3,\alpha\}$ and $\alpha$ is the affectance degradation distance.
The purpose of lower bound $h$ to $3$ is to isolate affectance among neighboring layers, as in the \RN model.
Otherwise, if $v$ is slow, it uses each time slot $t$ such that $t \equiv d(v)+h \pmod{2h}$.
Notice that this schedule separates transmissions that occur in the same time step as follows. For any pair of slow nodes, or pair of fast nodes, they are either at the same distance from the source or they are separated by at least $2h$ hops from each other. For any pair of one slow and one fast node, they are separated by at least $h$ hops. The reason to lower bound $h$ by $3$ is to avoid unnecessary interference between links separated by one hop.

The \defn{broadcasting schedule} for fast nodes is simple: upon receiving a packet for dissemination, transmit in the next time slot reserved. For slow nodes, the schedule is determined by a randomized contention resolution protocol that can be run in the reserved time slots. The protocol 
is simple: upon receiving a packet for dissemination, each slow node transmits repeatedly with probability $1/(4K(T_{\min},s))$, until the packet is delivered.

In the rest of this section, we bound the length of the broadcasting schedule. The following upper bound will be used.

\begin{lemma}
\label{lemma:maxrank}
The maximum rank of a LABST $T$ with source node $s$ is $$R(T)\leq \lceil M(T_{\min},s)\rceil.$$
\end{lemma}
\begin{proof}
Consider the construction of a LABST $T$ in Algorithm~\ref{alg:prunaff}. 
Consider any path from root to leaf in $T_{\min}$.
Because initially all nodes have rank $1$ (cf. Line~\ref{line:initrank}), by definition of $M(T,s)$, the total affectance on this path is at most $M(T_{\min},s)$.
Each time that a node in such path increases its rank in Line~\ref{line:increaserank}, the node becomes slow in Line~\ref{line:toslow}, and will not be fast again because the rank updates are carried level-by-level upwards the tree (cf. Line~\ref{line:upwards}). 
Thus, after the transformation, the claimed bound holds because a value $\geq 1$ is reduced from the total affectance due to fast nodes in the path (cf. Line~\ref{line:reduce}).
\end{proof}

\begin{theorem}
\label{thm:sched}
For any given network of $n$ nodes with a source node, diameter $D$, and affectance degradation distance $\alpha$, there exists a broadcasting schedule of length at most
$$D+2h\lceil M(T_{\min})\rceil(\lceil M(T_{\min})\rceil+16 K(T_{\min})\ln n),$$
where $h=\max\{3,\alpha\}$.
\end{theorem}
\begin{proof}
First we show that the broadcasting schedule is correct. 
Consider any pair of nodes $u,v \in V$ transmitting in the same time slot.
If $d(u)=d(v)$ and they are both fast nodes with the same rank, the affectance on each other's links is low by definition of the LABST.
If $d(u)=d(v)$ and they are both slow nodes, the contention resolution protocol will disseminate the packet to the next layer.
Otherwise, given the slot reservation, it is either $|d(u)-d(v)|= 2h$ if $u$ and $v$ are both fast or both slow, or $|d(u)-d(v)|= h$ if one is slow and the other fast. Given that $h\geq \alpha$, the affectance on each other's links is negligible. 

To prove the schedule length, consider any path $p$ from root to leaf in the LABST $T$. The path $p$ can be partitioned into consecutive maximal subpaths according to rank. In each maximal subpath $p'\in p$ of consecutive nodes of the same rank, the first node may have to wait up to $2hR(T)$ slots for the next reserved time slot, but after that all nodes except the last one transmit in consecutive time slots. Given that there are at most $R(T)$ such maximal subpaths and that their aggregated length is at most $D(T)$, the schedule length in the fast nodes of path $p$ is at most $D(T)+2hR(T)^2\leq D+2hR(T)^2$, where the latter inequality holds because $T$ is a BFS tree.

Consider now any link $\ell\in p$ where the rank changes, that is $r(p(\ell))\neq r(c(\ell))$ and $p(\ell)\in S_{d(p(\ell))}\subseteq V_{d(p(\ell))}$. Recall that the schedule in such link is defined by a randomized contention resolution protocol where each node transmits with probability $1/(4K(T_{\min}))$, where
$$K(T_{\min}) = \max_d \max_{V'\subseteq V_d(T_{\min})} \frac{1}{|L(V')|} a_{V'} (L(V')),$$
where $L(V')$ is the set of BFS tree links between $V'$ and nodes at distance $d+1$ of the source,
and $V_d(T_{\min})$ is the set of nodes at distance $d$ from the source in $T_{\min}$.
For a probability of transmission $$q\leq \frac{1}{4 \max_{S\subseteq V_{d(p(\ell))}}  a_{S} (L(S))/|L(S)|},$$ 
it was proved in~\cite{KVaff} that
the probability that 
there is still some link in $S$ where no transmission was successful
after $4c\ln |V_{d(p(\ell))}|/q$ time slots running Algorithm 1 in~\cite{KVaff}, 
is at most $|V_{d(p(\ell))}|^{1-c}$, $c>1$. 
Given that $1/(4K(T_{\min}))$ verifies such condition, we know that after $$16cK(T_{\min})\ln |V_{d(p(\ell))}| \leq 16cK(T_{\min})\ln n$$ (reserved) time slots, the transmission in link $\ell$ has been successful with positive probability. 
%at least $1-|V_{d(p(\ell))}|^{1-c}$, $c>1$.
Given that there are at most $R(T)-1$ links where the rank changes, using the union bound, we know that after $(R(T)-1)16cK(T_{\min})\ln n$ (reserved) time slots all slow nodes have delivered their packets with some positive probability, which shows the existence of a deterministic schedule of such length\footnote{In settings with collision detection and where the affectance on any given link is $O(n)$, a big enough constant $c>1$ yields a randomized protocol that succeeds with probability $1-1/n$.}. 
The time slots reserved for slow nodes appear with a frequency of $2h$. Thus, the schedule length in the slow nodes of path $p$ is at most $2h(R(T)-1)16cK(T_{\min})\ln n\leq 32h R(T)K(T_{\min})\ln n$, for $c= R(T)/(R(T)-1)$.

Adding both schedule lengths we have
$$D+2hR(T)^2+32h R(T)K(T_{\min})\ln n$$
Replacing the bound on $R(T)$ in Lemma~\ref{lemma:maxrank}, the claim follows.
\end{proof}

For networks with affectance degradation distance $\lceil\log n\rceil$, Theorem~\ref{thm:sched} yields the following corollary.
\begin{corollary}
\label{c:broadcast}
For any given network of $n\geq 8$ nodes, diameter $D$, and affectance degradation distance $\lceil\log n\rceil$, there exists a broadcasting schedule of length 
$$D+O\left(\log^2 n M(T_{\min}) \left(  \frac{M(T_{\min})}{\log n} + K(T_{\min}) \right)\right).$$
\end{corollary}

For comparison, for less contentious networks where affectance at more than one hop is not present (\RN model), using a GBST a broadcast schedule of length $D+O(\log^3 n)$ was shown in~\cite{GPQ:gossipTree} and of
length $O(D+\log^2 n)$ was proved in~\cite{KowalskiP07}.

%!TEX root = ./KMZ_mmb.tex

\section{A Dynamic \mmb Protocol}
\label{section:prot}

In this section, we present our \mmb protocol and we bound its competitive throughput. The protocol uses the LABST\footnote{We refer to the tree and the broadcast schedule indistinctively.} 
presented in Section~\ref{sec:BT}.\footnote{%
Any broadcast schedule that works under the affectance model could be used.%
} 
The intuition of the protocol is the following. Each source node has a (possibly empty) queue of packets that have been injected for dissemination. Then, starting with an arbitrary source node $s\in S$ with ``large enough'' number of packets in its queue, packets are disseminated through a LABST rooted at $s$. If the number of packets in the queue of $s$ becomes ``small'', $s$ stops sending packets and, after some delay to clear the network, another source node $s'\in S$ starts disseminating packets through a LABST rooted at $s'$. The procedure is repeated following the order of a list of source nodes, which is dynamically updated according to queue sizes to guarantee good throughput.
Packets from any given source are pipelined with some delay to avoid collisions and affectance. 
Being a radio broadcast network, packets might be received earlier than expected using links or time slots other than those defined by the LABST. If that is the case, to guarantee the pipelining, nodes ignore those packets. 

The following notation will be also used. 
The LABST rooted at $s\in S$ is denoted as $T(s)$.
We denote the length of the broadcast schedule (time to deliver to all nodes) from $s$ as $\Delta(s)$, and $\Delta=\max_{s\in S} \Delta(s)$.
Let the pipeline delay (the time separation needed between consecutive packets to avoid collisions and affectance) from $s$ be $\delta(s)$, and $\delta=\max_{s\in S} \delta(s)$.
Given a node $i\in S$ and time slot $t$, the length of the queue of $i$ is denoted $\ell(i,t)$.
Let the length of all queues at time $t$ be $\ell(t)=\sum_{i\in S} \ell(i,t)$.
We say that, at time $t$, a node $i$ is
\defn{empty} if $\ell(i,t)<\Delta$,
\defn{small} if $\Delta\leq\ell(i,t)<n\Delta$, and
\defn{big} if $\ell(i,t)\geq n\Delta$.

Consider the following \defn{\mmb Protocol}.

\begin{enumerate}
\item For each source node $s\in S$ define a LABST rooted at $s$.
\item Define a Move-big-to-front (MBTF) list~\cite{ChlebusKR09} of source nodes, initially in any order. According to this list, source nodes circulate a token. While being disseminated, the token has a time-to-live counter of $\Delta$, maintained by all nodes relaying the token. A source node $s$ receiving the token has to wait for the token counter to reach zero before starting a new transmission. Let the time slot when the counter reaches zero be $t$. Then, node $s$ does the following depending on the length of its queue.
\begin{enumerate}
\item If $s$ is empty at $t$, it passes the token to the next node in the list. We call this event a \defn{silent round}. 
\item If $s$ is small at $t$, it broadcasts $\Delta$ packets pipelining them in intervals of $\delta$ slots. After $\delta$ more slots, it passes the token to the next node in the list. 
\item If $s$ is big at $t$, it moves itself to the front of the list. We call this event a \defn{discovery}. Then, $s$ broadcasts packets pipelining them in intervals of $\delta$ slots as long as it is big, but a minimum of $\Delta$ packets. With the first of these packets $s$ broadcasts the changes in the list. $\delta$ more slots after transmitting these packets, it passes the token to the next node in the list.
\end{enumerate}
\end{enumerate}

The following theorem shows an upper bound on the number of packets in the system at any time, which allows to prove the competitive throughput of our protocol. The proof structure is similar to the proof in~\cite{ChlebusKR09} for MBTF, but many details have been redone to adapt it to a multihop network.

\begin{theorem}
\thmlabel{pipeline}
For any given network of $n$ nodes,
at any given time slot $t$ of the execution of the \mmb protocol defined, 
the overall number of packets in queues is $\ell(t)< (t\delta/(1+\delta))+2\Delta n^2$.
\end{theorem}

\begin{proof}
For the sake of contradiction, assume that there exists a time $t$ such that the overall number of packets in the system is $\ell(t)\geq (t\delta/(1+\delta))+2\Delta n^2$. The number of packets in queues at the end of any given period of time is at most the number of packets in queues at the beginning of such period, plus the number of time slots when no packet is delivered, given that at most one packet is injected in each time slot. We arrive to a contradiction by upper bounding the number of time slots when no packet is delivered within a conveniently defined period before $t$. 
Consider the period of time $T$ such that
\begin{align}
%T &= c\Delta(1+\delta) \textrm{ for some constant $c>0$}\\
\ell(t-T) &\leq n^2\Delta+\frac{(t-T)\delta}{1+\delta} \eqlabel{ub}\\
\forall t'\in[t-T,t] : \ell(t') &\geq n^2\Delta \eqlabel{lb}\\
\ell(t) &\geq (t\delta/(1+\delta))+2\Delta n^2 \eqlabel{endT}
\end{align}
%Consider the same definitions of token-pass and life cycle used in~\cite{ChlebusKR09}. 
From now on, the analysis refers to the period of time $T$. We omit to specify it for clarity.
Let $C\subseteq S$ be the set of nodes that are big at some point. % in $T$. 
Due to the pigeonhole principle and Equation~\eqref{lb}, we know that for each time slot %in $T$ 
there is at least one big source node. In other words, the token cannot be passed throughout the whole list without at least one discovery.
As a worst case, assume that only nodes in $C$ have packets to transmit. %in $T$.
For each node $i\in C$, the token has to be passed through at most $|S\setminus C|\leq n-|C|$ nodes that are not in $C$ before $i$ is discovered, because after $i$ is discovered no node in $S\setminus C$ will be before $i$ in the list. 
%Call this a silent round (remember that now a round includes handshaking the token).
Hence, there are at most $|C|(n-|C|)$ silent rounds, each of length $\Delta$ for token pass.
So, due to passing the token through nodes in $S\setminus C$, there are at most $|C|(n-|C|) \Delta$ time slots when no packet is delivered.

We bound now the time slots when no packet is delivered due to passing the token through nodes in $C$ before being discovered for the first time. %in $T$. 
Consider any given node $i\in C$. The argument is similar to the previous case. Any other node $j\in C$ that is discovered before $i$ is moved to the front of the list. If $i$ is going to be before $j$ in the list later, it is not going to happen before $i$ is discovered for the first time. 
Then, before $i$ is discovered, %in $T$, 
it may hold the token at most $|C|-1$ times. As a worst case, assume that for each of these times $i$ is empty.
Hence, there are at most $|C|(|C|-1)$ silent rounds, each of length $\Delta$ for token pass.
So, due to passing the token through nodes in $C$ before being discovered, %in $T$, 
there are at most $|C|(|C|-1) \Delta$ time slots when no packet is delivered.

It remains to bound the time slots when no packet is delivered due to pipelining and passing the token through nodes in $C$ after being discovered. %in $T$.
Consider any given node $i\in C$ after being discovered. 
%in $T$. 
If $i$ is big during the rest of $T$, 
it broadcasts packets pipelining them in intervals of $\delta$ slots. If instead $i$ becomes small during $T$, 
$i$ will have $\Delta$ packets to transmit for at least $n-1$ times that holds the token afterwards before becoming empty, because right after becoming small it has at least $(n-1)\Delta$ packets in queue. And there are at most $n-1$ nodes in $C$ that will not be behind $i$ in the list until $i$ becomes big again. Hence, $i$ always has $\Delta$ packets to transmit after being discovered the first time. After becoming small, $i$ has to pass the token to the next node in the list introducing a delay of $\Delta$. As a worst case scenario, we assume that upon each discovery of each node $i\in C$, only $\Delta$ packets are broadcast before passing the token. Then, for each $\Delta$ packets delivered, 
there are at most $\Delta+\Delta(\delta-1)=\Delta\delta$ time slots when no packet is delivered, over a period of $\Delta+\Delta\delta = \Delta(1+\delta)$ time slots.
Because $C$ is the set of nodes that are discovered in $T$, we can bound the number of batches of $\Delta$ packets delivered in $T$ by $\lfloor T/(\Delta(1+\delta))\rfloor \leq T/(\Delta(1+\delta))$.
Then, there are at most $T\Delta\delta/(\Delta(1+\delta))  = T\delta/(1+\delta)$ time slots when no packet is delivered due to nodes in $C$ after being discovered.

Combining these bounds with Equation~\eqref{ub}, we have that there are at most
\begin{align*}
n^2\Delta+\frac{(t-T)\delta}{1+\delta} &+ |C|(n-|C|) \Delta+|C|(|C|-1) \Delta + \frac{T\delta}{1+\delta}\\
&= n^2\Delta+\frac{t\delta}{1+\delta}+\Delta|C|(n-1)\\
&< \frac{t\delta}{1+\delta}+2\Delta n^2
\end{align*}
 time slots when no packet is delivered. Which is a contradiction.
%\qed
\end{proof}

\begin{lemma}
\label{lemma:convergence}
%For any given network of $n$ nodes, diameter $D$, and affectance degradation distance $\alpha$,
There exists a \mmb protocol that achieves a competitive throughput of at least 
$$\lim_{t\to\infty} \frac{1}{1+\delta}-\frac{2\Delta n^2}{t}.$$
%Where $\delta=\max\{3,\alpha\}16\overline{A}\ln n$, 
%and $\Delta\leq D + 2h(2\lceil\log n\rceil + \overline{A})^2 + 32h\overline{A} (2\lceil\log n\rceil + \overline{A})\ln (n-1)$.
\end{lemma}

\begin{proof}
A packet is delivered when it has been received by \emph{all} nodes.
The optimal algorithm delivers at most one packet per time slot, since any given node can receive at most one packet per time slot. 
Additionally, the injection is limited to be feasible, that is, there must exist an optimal algorithm OPT such that the latency of each packet is bounded for OPT. Thus, at most one packet may be injected in each time slot.
Then, the competitive throughput is at least 
\begin{align*}
\lim_{t\to\infty}\frac{d_{ALG}(t)}{d_{OPT}(t)} \geq
\lim_{t\to\infty} \frac{t-nd_{ALG}(t)}{t},
\end{align*}
where $nd_{ALG}(t)$ is the max number of packets that could not be delivered by ALG by time $t$.
Using the bound in~\thmref{pipeline} we have that
\begin{align*}
\lim_{t\to\infty}\frac{d_{ALG}(t)}{d_{OPT}(t)} 
&\geq \lim_{t\to\infty}\frac{t-(t\delta/(1+\delta))-2\Delta n^2}{t}\\ 
&\geq \lim_{t\to\infty} \frac{1}{1+\delta}-\frac{2\Delta n^2}{t}. 
\end{align*}
%\qed
\end{proof}

The following theorem shows our main result.
\begin{theorem}
\label{thm:main}
For any given network of $n$ nodes, diameter $D$, and affectance degradation distance $\alpha$,
there exists a \mmb protocol that achieves a competitive throughput of at least 
$$\lim_{t\to\infty} \frac{1}{1+\delta}-\frac{2\Delta n^2}{t}.$$
Where 
$\Delta \leq D+2h\lceil M\rceil(\lceil M\rceil+16 K\ln n)$,
$h = \max\{3,\alpha\}$,
$K = \max_{s\in S}K(T_{\min}(s),s)$,
$M = \max_{s\in S}M(T_{\min}(s),s)$,
and $\delta \leq 16hK\ln n$.

\end{theorem}
\begin{proof}
The length $\Delta(s)\leq \Delta$ of the broadcast schedule in a LABST rooted at $s$ is given in \thmref{sched}.
With respect to $\delta(s)\leq \delta$, as explained in the proof of \thmref{sched}, slow nodes at distance $d$ from the root deliver a packet to the next node in a path of a LABST $T(s)$ within $16cK(T_{\min}(s))\ln |V_d|$ with positive probability for any $c>1$. This shows the existence of a deterministic schedule of that length. Additionally, packets must be separated by at least $\max\{3,\alpha\}$ to avoid collisions and affectance from nodes at different distances from the source (see the proof of \thmref{sched} for further details). Then, it is $\delta(s)=\max\{3,\alpha\}16K(T_{\min}(s))\ln n$, for $c=\ln n / \ln|V_d|$. Replacing, the claim follows.
%\qed
\end{proof}

The above theorem yields the following corollary that provides intuition. 
\begin{corollary}
For any given network of $n$ nodes, diameter $D$, and affectance degradation distance $\alpha$,
there exists a \mmb protocol such that the competitive throughput converges to
%$$ \frac{1}{1+16(3+\alpha)\ln n\max_{s\in S}K(T_{\min}(s),s)}.$$
$$\frac{1}{O(\alpha K\log n)},$$
where $K=\max_{s\in S}K(T_{\min}(s),s)$.
\end{corollary}

To evaluate these results, it is important to notice that the competitive throughput bound was computed against a theoretical optimal protocol that delivers one packet per time slot, which is not possible in practice in a multi-hop network. For comparison, instantiating our interference model in the \RN model (no affectance), using the WEB protocol~\cite{CW:waveexp} for slow transmissions our \mmb protocol can be shown to converge to $1/O(\log^2 n)$. 
Furthermore, for single-instance multi-broadcast in \RN, Ghaffari et al. showed in~\cite{GHK:throughputLB} a throughput upper bound of $O(1/\log n)$ for any algorithm. Although this bound is worst-case, it can be compared with our $1/O(\alpha K \log n)$ that applies even under affectance.
%We evaluate the \RN case through simulations of our protocol in the following section.

%!TEX root = ./KMZ_algorithmica.tex

\section{Simulations}
\label{s:simul}
We carried out simulations of our \mmb protocol under the affectance model. 

For each of three values of $n=16,32,64$, we produced three types of input networks to study the impact of high, medium and low interference. Namely, (1) a complete bipartite graph of $n$ nodes split evenly (high interference), (2) two overlapped trees obtained running BFS on a connected random graph from different nodes (medium interference), and (3) a single path of nodes (low interference). 
For these inputs, we set the affectance degradation distance $\alpha$ to $\max\{1,\sqrt{\log n}\}$, $\max\{1,(\log n)/2\}$, and $\max\{1,\log n\}$ respectively.
Recall that our model of general affectance comprises any interference effect (as long as it is additive). Hence, rather than restricting to a specific model for our simulations (such as SINR or \RN models which are geometric), we produced a less restrictive affectance matrix that comprises only the effect of the distance in hops from the interfering node to the receiver of the link.
Specifically, for each input graph $G=\{V,E\}$, we computed an affectance matrix as follows. Let $d(i,k)$ be the shortest distance in hops from node $i$ to node $k$. Then, for each node $i\in V$ and directed link $(j,k)\in E$,
\begin{align*}
A(i,(j,k)) &= \left\{
\begin{array}{ll}
0 & \textrm{if $d(i,k)\geq\alpha$}\\ 
1 & \textrm{if $d(i,k)=0$}\\ 
1/d(i,k)^2 & \textrm{if $0<d(i,k)<\alpha$}\\ 
\end{array} \right.
\end{align*}

To simulate our \mmb protocol, we first computed the broadcast schedule from each node (although later we select only a subset of nodes for injection).
% (cf. Section~\ref{sec:BT}).
For each node $i$, we computed a BFS tree rooted at $i$, we ranked the tree according to affectance, and we computed the average tree-layer affectance. To ensure polynomial running time, we computed \emph{some} BFS tree rooted at each node $i$, rather than aiming for the tree that minimizes the polynomial on the affectance metrics (cf. Section~\ref{sec:BT})
which would require to compute \emph{all} BFS trees rooted at $i$. For the same reason, we computed $K(T,i)=\max_d  a_{V_d(T)} L(V_d(T))/|L(V_d(T))|$, rather than aiming in each layer for the subset of tree links that maximizes the average affectance, which would require to compute the average affectance of \emph{all} subsets of links. 

For the input networks produced uniformly at random the above simplification should not have a significant impact in performance. Also, taking the tree-layer average affectance of all nodes in the layer guarantees that the contention resolution protocol used for slow nodes disseminates the packet to the next layer also under the affectance model, within the time bounds specified in the proof of Theorem~\ref{thm:sched}, since interference from \emph{all} nodes in the layer is taken into account.

With respect to the set of source nodes $S$, each node was chosen to be a source at random with probability $1/3$.
For each source $s$, we computed the length of the schedule (time to deliver to all other nodes) as $\Delta(s)=D(T)+2hR(T)^2+32h R(T)K(T)\ln n$, and the pipeline delay (time separation needed between consecutive packets) as $\delta(s)=16hK(T)\ln n$.
Notice that we used $D(T)$ and $R(T)$ rather than their bounds in Theorem~\ref{thm:sched}, since for the simulations we know their values.
Finally we computed $\Delta=\max_{s\in S}\{\Delta(s)\}$ and $\delta=\max_{s\in S}\{\delta(s)\}$.

The queue of one source node was initialized to $2\Delta\delta$ packets, and the rest was left empty. That is, initially there is one big node and the rest are empty, introducing overhead due to token passing through future injections.
To evaluate performance, packets were injected at different rates and with different policies. Specifically, we tested injection rates $1$ (feasibility upper bound), $1/\sqrt{\delta}$ and $1/\delta$ (approximate theoretical guarantee on delivery rate). Target source nodes for injections were chosen with four different policies: (1) uniformly among source nodes, (2) next (according to ID) source node after the node currently delivering, (3) always in the current source node, and (4) uniformly among all source nodes except the node currently delivering. The idea for policy (3) was to evaluate the system when the token does not circulate, whereas policy (4) can be seen as a worst-case injection since then the likelihood of having one node being big forever is low.
 
The results of the simulations for the second type of input graph (two overlapped trees obtained running BFS on a connected random graph from different nodes) are illustrated in Figures~\ref{fig:TvsI16} to~\ref{fig:Tvsdelta}. Similar results were obtained for the other two types of input graph. 

Figures~\ref{fig:TvsI16}, \ref{fig:TvsI32}, and~\ref{fig:TvsI64} show the competitive throughput as a function of the number of packets injected for $n=16, 32$, and $64$ respectively. In each of these figures, the 12 combinations of injection rates and policies described above are shown. It can be seen in these plots that the competitive throughput converges to a value $1/(1+\delta)$ when the injection rate is $1$, and much higher for smaller injection rates. 

In some cases when the injection rate is low, after an initial phase when it converges to high values, the competitive throughput is reduced due to overhead from token passing. Nevertheless, 
it can be seen in Figures~\ref{fig:TvsI16} and~\ref{fig:TvsI32} that still it converges to a value that is $1/(1+\delta)$ or higher.
For further illustration of these observations, we include a plot of competitive throughput versus $1+\delta$ in Figure~\ref{fig:Tvsdelta}. In this plot, the competitive throughput obtained for each combination of network size, injection rate and policy is compared with the $1/(1+\delta)$ lower bound proved in Theorem~\ref{thm:main}. As in previous plots, it can be seen that for the practical scenarios evaluated, our \mmb protocol behaves as shown in our analysis or better.

It is important to notice that the competitive throughput was computed against a theoretical optimal protocol that delivers a packet immediately after injection, which is not possible in practice in a multi-hop network.

\begin{figure*}[htb]
\begin{center}
\resizebox{\textwidth}{!}{\input{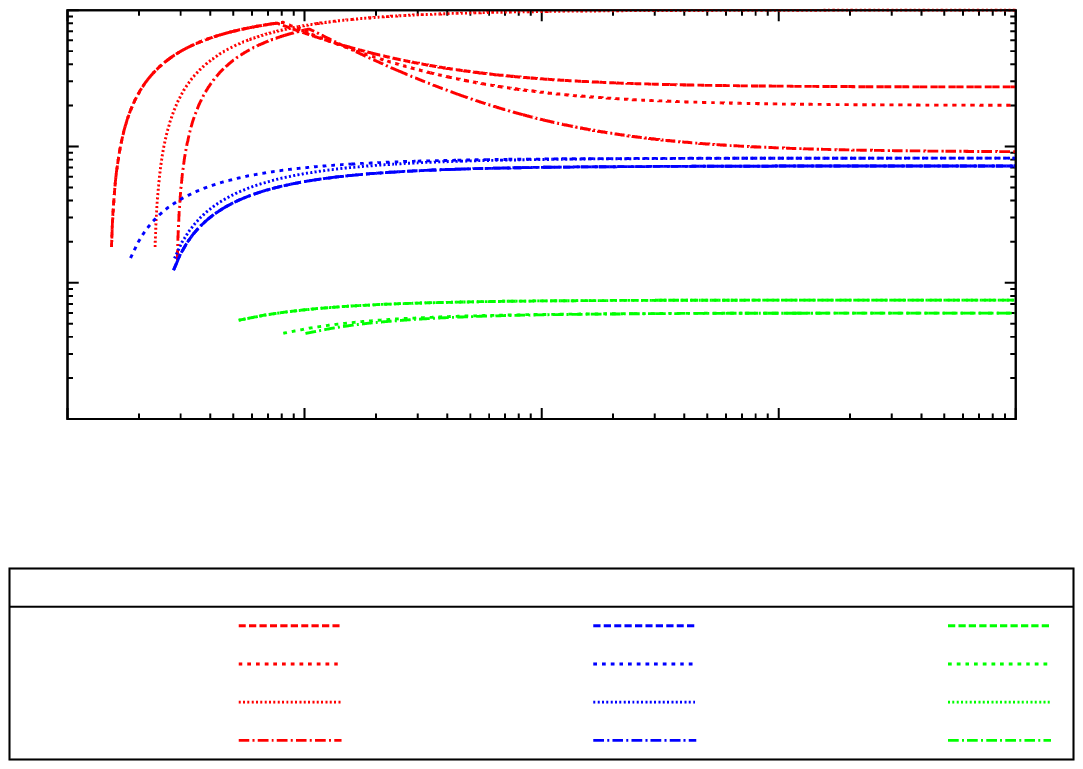}}
\caption{Competitive throughput vs. packets injected. $n=16$}
\label{fig:TvsI16}
\end{center}
\end{figure*}
\begin{figure*}[htb]
\begin{center}
\resizebox{\textwidth}{!}{\input{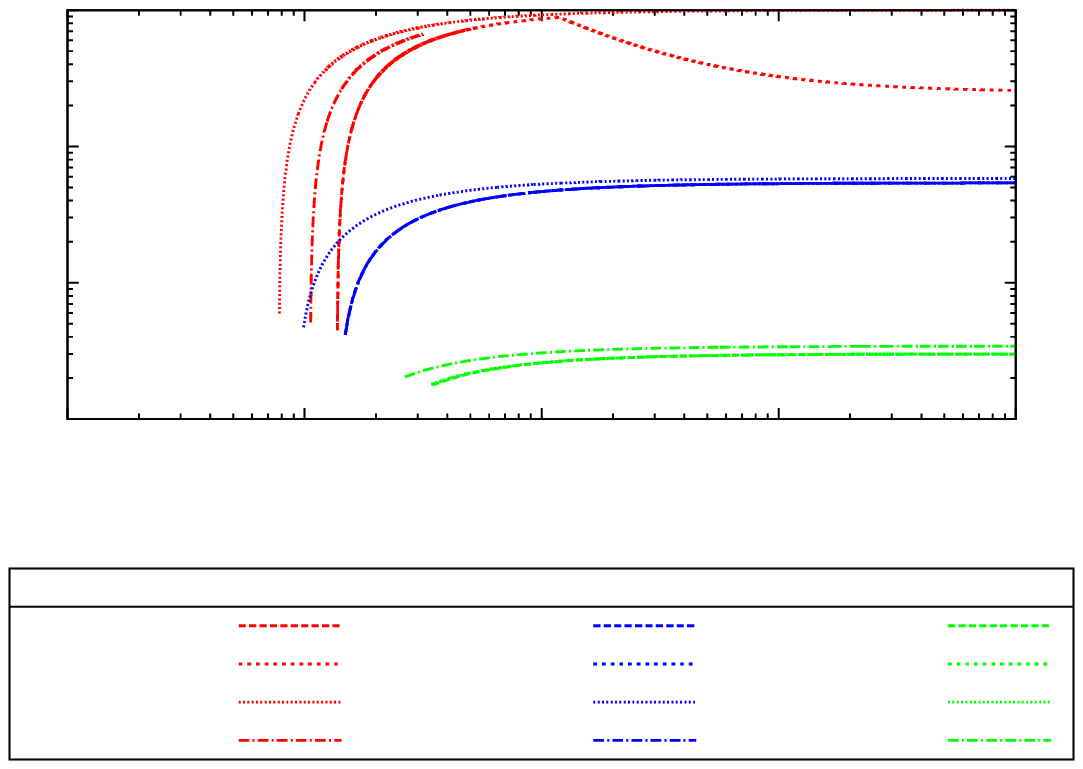}}
\caption{Competitive throughput vs. packets injected. $n=32$}
\label{fig:TvsI32}
\end{center}
\end{figure*}
\begin{figure*}[htb]
\begin{center}
\resizebox{\textwidth}{!}{\input{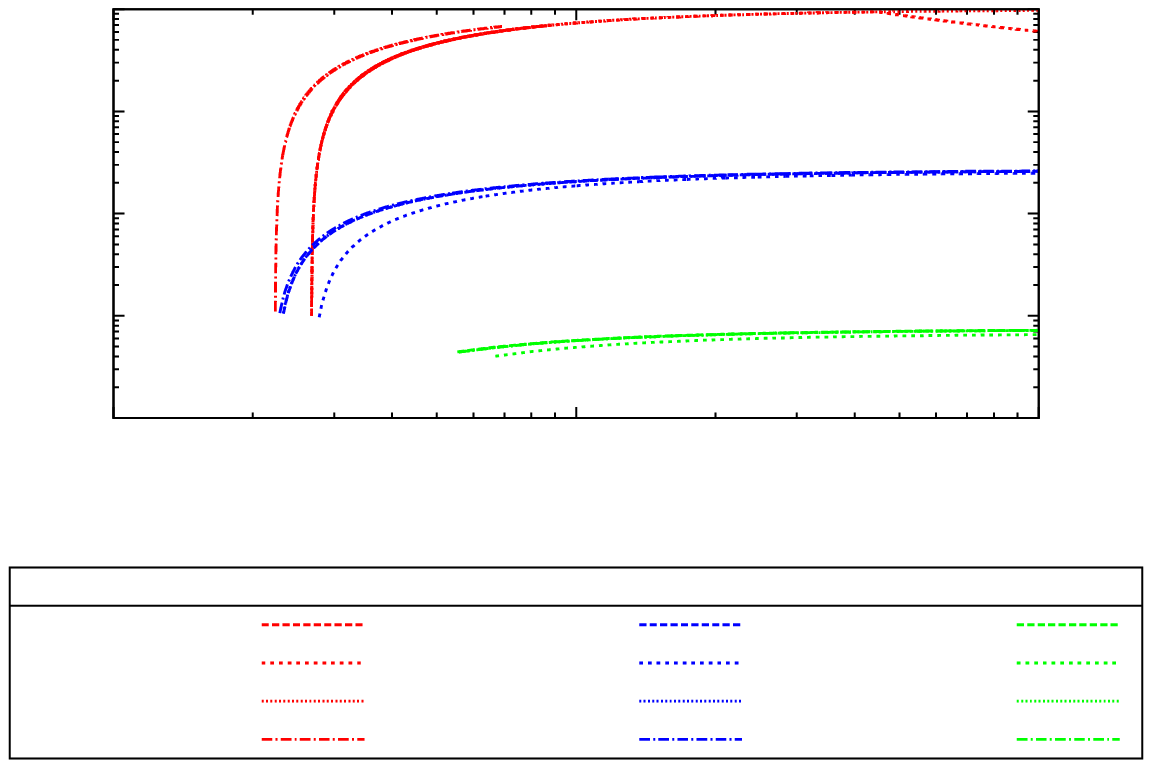}}
\caption{Competitive throughput vs. packets injected. $n=64$}
\label{fig:TvsI64}
\end{center}
\end{figure*}
\begin{figure*}[htb]
\begin{center}
\resizebox{\textwidth}{!}{\input{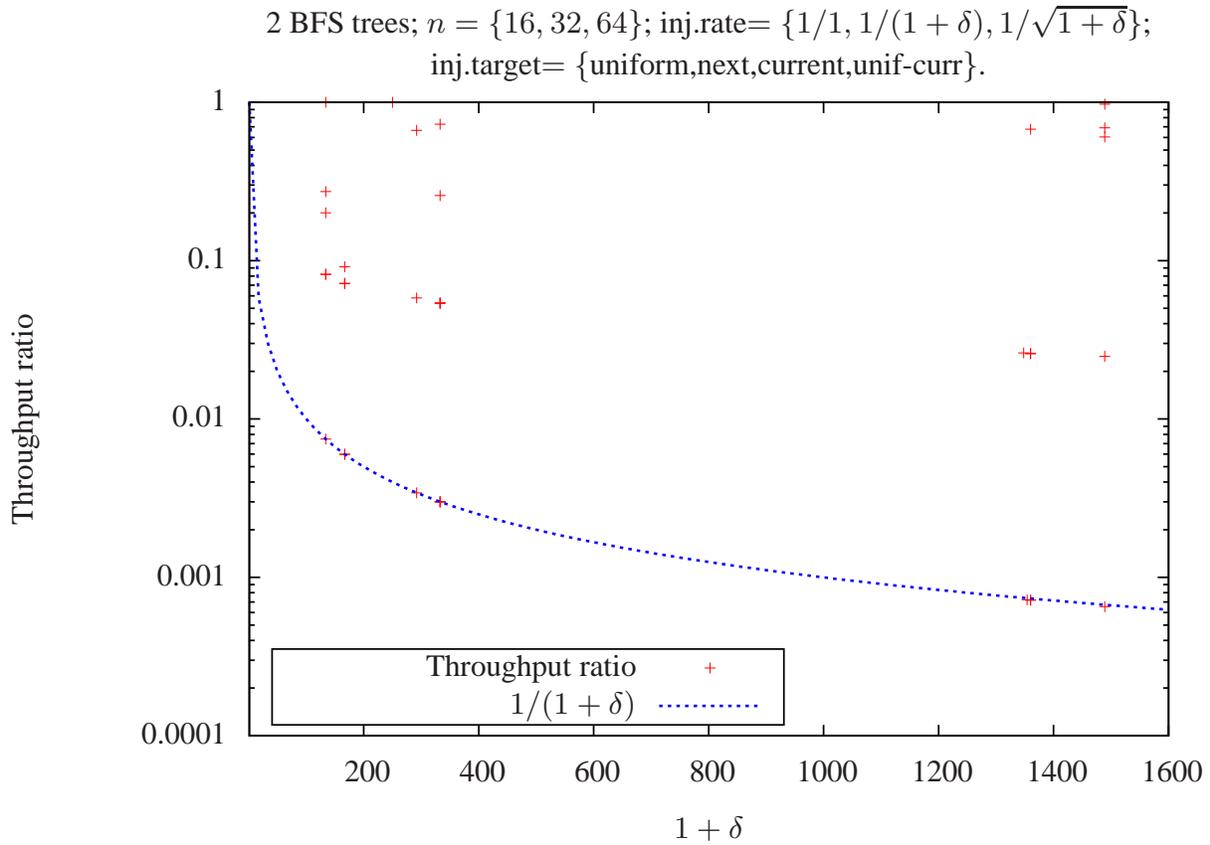}}
\caption{Competitive throughput vs. $1+\delta$.}
\label{fig:Tvsdelta}
\end{center}
\end{figure*}

\bibliographystyle{abbrv}
\bibliography{references,references2}

\appendix
%!TEX root = ./KMZ_mmb.tex

\section*{Appendix}

\section{Modeling SINR under the Affectance Model}
\label{s:mapping}

\begin{claim}
The affectance matrix 
\begin{align*}
A(w,(u,v)) &= 
\left\{ \begin{array}{ll}
 0 & \textrm{if $w=u$},\\
 \frac{P/d_{wv}^\alpha}{P/(\beta' d_{uv}^\alpha)-N} & \textrm{otherwise}.
  \end{array} \right.
\end{align*}
corresponds to the SINR model. 
\end{claim}

\begin{proof}
To prove this claim, we show that there is a successful transmission in the SINR model if and only if there is a successful transmission in the affectance model with matrix $A$. 

Consider a successful transmission in the SINR model. We have 
\begin{align*}
\frac{P/d_{uv}^\alpha}{N+\sum_{w\neq u} P/d_{wv}^\alpha} &> \beta'\\
P/(\beta' d_{uv}^\alpha) &> N+\sum_{w\neq u} P/d_{wv}^\alpha\\
P/(\beta' d_{uv}^\alpha) - N &> \sum_{w\neq u} P/d_{wv}^\alpha.
\end{align*}
If $\sum_{w\neq u} P/d_{wv}^\alpha=0$ then $\sum_{w\neq u}A(w,(u,v))=0\Rightarrow$ success in affectance model. Otherwise, it is $\sum_{w\neq u} P/d_{wv}^\alpha>0$ and we have
\begin{align*}
\frac{P/(\beta' d_{uv}^\alpha) - N}{\sum_{w\neq u} P/d_{wv}^\alpha} &> 1.
\end{align*}
Thus, it is $P/(\beta' d_{uv}^\alpha) - N>0$ and, hence, we have
\begin{align*}
\frac{\sum_{w\neq u} P/d_{wv}^\alpha}{P/(\beta' d_{uv}^\alpha)-N} &< 1.
\end{align*}
Therefore, it is $\sum_{w\neq u} A(w,(u,v)) < 1 \Rightarrow$ success in affectance model.

Consider now a non-successful transmission in the SINR model. We have 
\begin{align*}
\frac{P/d_{uv}^\alpha}{N+\sum_{w\neq u} P/d_{wv}^\alpha} &\leq \beta'\\
P/(\beta' d_{uv}^\alpha) &\leq N+\sum_{w\neq u} P/d_{wv}^\alpha\\
P/(\beta' d_{uv}^\alpha) - N &\leq \sum_{w\neq u} P/d_{wv}^\alpha.
\end{align*}
If $P/(\beta' d_{uv}^\alpha) \leq N$, it would mean that $P$ is not large enough to overcome the noise, even if no other node transmits. Then, rather than being produced by interference, the failure would be due to consider a link that is not even feasible. That is, $(u,v)\notin E$. Thus, it must be $P/(\beta' d_{uv}^\alpha) > N$ and we have
\begin{align*}
\frac{\sum_{w\neq u} P/d_{wv}^\alpha}{P/(\beta' d_{uv}^\alpha)-N} &\geq 1\\
\end{align*}
Therefore, it is
$\sum_{w\neq u}A(w,(u,v)) \geq 1 \Rightarrow$ failure in affectance model.
\end{proof}

%!TEX root = ./KMZ_mmb.tex

\section{Notes}
\label{sec:diff}

In this section, we highlight the differences between this paper and our preliminary work appeared in~\cite{KMR_fomc14}.

In~\cite{KMR_fomc14}, we studied a model of affectance that subsumes only \emph{some} SINR models, by combining the effect of \RN collisions with affectance from nodes at more than one hop. Here, we generalize our model to subsume any arbitrary interference model. For instance, in the present model it is possible to receive a transmission even when more than one neighboring node transmits, as in some SINR models.

Also, in~\cite{KMR_fomc14} our maximum path affectance metric was based on fast links only, which yields possibly tighter bounds. However, the definition was based on a specific BFS tree (a GBST~\cite{GPQ:gossipTree}) which related the network characterization to our specific algorithmic solution. In the present work the characterization is related only to topology, since it is based on arbitrary BFS trees.

We also notice here that the proof of the maximum rank in~\cite{KMR_fomc14} has an error, introduced while bounding the maximum number of ranks needed for updating the rank according to affectance. Lemma~\ref{lemma:maxrank} here provides the correct bound.

\end{document}